\documentclass[11pt, draftcls, onecolumn]{IEEEtran}
\usepackage{subfigure}
\usepackage{setspace}
\usepackage{amsmath}
\usepackage{amssymb}
\usepackage{amsfonts}
\usepackage{amscd}
\usepackage{mathrsfs}
\usepackage[final]{graphicx}
\usepackage{graphicx}
\usepackage{psfrag}
\usepackage{psfig}
\usepackage{color}
\usepackage{url}
\newtheorem{theorem}{Theorem}
\newtheorem{lemma}{Lemma}
\newtheorem{definition}{Definition}

\begin{document}

\title{ Low ML-Decoding Complexity, Large Coding Gain, Full-Rate, Full-Diversity STBCs for  $2\times2$ and $4\times2$ MIMO Systems}

\author{
\authorblockN{K. Pavan Srinath and B. Sundar Rajan\\
Dept of ECE, Indian Institute of science \\
Bangalore 560012, India\\
Email:\{pavan,bsrajan\}@ece.iisc.ernet.in\\
}
}
\maketitle
\begin{abstract}
This paper\footnote[1]{Part of the content of this manuscript has been accepted for presentation in IEEE Globecom 2008, to be held in New Orleans} deals with low maximum likelihood (ML) decoding complexity, full-rate and full-diversity space-time block codes (STBCs), which also offer large coding gain, for the  2 transmit antenna, 2 receive antenna ($2\times 2$) and the 4 transmit antenna, 2 receive antenna ($4\times 2$) MIMO systems. Presently, the best known STBC for the $2\times2$ system is the Golden code and that for the $4\times2$ system is the DjABBA code. Following the approach by Biglieri, Hong and Viterbo, a new STBC is presented in this paper for the $2\times 2$ system. This code matches the Golden code in performance and ML-decoding complexity for square QAM constellations while it has lower ML-decoding complexity with the same performance for non-rectangular QAM constellations. This code is also shown to be \emph{information-lossless} and \emph{diversity-multiplexing gain} (DMG) tradeoff optimal. This design procedure is then extended to the $4\times 2$ system and a code, which outperforms the DjABBA code for QAM constellations with lower ML-decoding complexity, is presented. So far, the Golden code has been reported to have an ML-decoding complexity of the order of $M^4$ for square QAM of size $M$. In this paper, a scheme that reduces its ML-decoding complexity to $M^2\sqrt{M}$ is presented.
\end{abstract}

\section{Introduction And Background}
Multiple-input multiple-output (MIMO) transmission has attracted a lot of interest in the last decade, chiefly because of the enhanced capacity it provides compared with that provided by the single-input, single-output (SISO) system. The Alamouti code \cite{SMA} for two transmit antennas, due to its orthogonality property, allows a low complexity ML-decoder. This scheme led to the development of the generalized orthogonal designs \cite{TJC}. These designs are famous for the simplified ML-decoding that they provide. They allow all the symbols to be decoupled from one another and hence, are said to be single-symbol decodable. Another bright aspect about these codes is that they have full transmit diversity for arbitrary complex constellation. However, the limiting factor of these designs is the low code rate (refer Section \ref{sysmod} for a definition of code rate) that they support.

At the other extreme are the well known codes from division algebra, first introduced in \cite{SRS}. The well known perfect codes \cite{ORBV} have also been evolved from division algebra with large coding gains. These codes have full transmit diversity  and have the advantage of a very high symbol rate, equal to that of the VBLAST scheme, which, incidentally doesn't have full transmit diversity. But unfortunately, the codes from division algebra including perfect codes have a very high ML-decoding complexity (refer Section \ref{sysmod} for a definition of ML-decoding complexity), making their use prohibitive in practice.

The class of single-symbol decodable codes also includes the codes constructed using co-ordinate interleaving, called co-ordinate interleaved orthogonal designs (CIODs) \cite{ZS}, and the Clifford-Unitary Weight single-symbol decodable designs (CUW-SSD) \cite{Sanjay}. These designs allow a symbol rate higher than that of the orthogonal designs, although not as much as that provided by the codes from division algebra. The disadvantage with these codes when compared with the Orthogonal designs is that they have full transmit diversity for only specific complex constellations.

The Golden code \cite{BRV}, developed from division algebra, is a full-rate (see Section \ref{sysmod} for the definition of full-rate), full-diversity $2\times2$ code for integer lattice constellations, but has been known to have a high ML-decoding complexity, of the order of $M^4$, where $M$ is the size of the constellation used (it is shown in Section \ref{mlexist} that this can be reduced significantly to $M^2\sqrt{M}$ when the constellation employed is a square QAM). It has to be mentioned that the codes presented in \cite{DV} and \cite{JJK} also have the same coding gain and ML-decoding complexity as the Golden code does. With a view of reducing the ML-decoding complexity, two new full-rate, full-diversity codes for QAM constellations have been proposed for the $2\times2$ MIMO system. The first code was independently discovered by Hottinen, Tirkkonen and Wichman \cite{HTW} and by Paredes, Gershman and Alkhansari  \cite{PGA}, which we call the HTW-PGA code and the second, which we call the Sezginer-Sari code, was reported in \cite{SS} by Sezginer and Sari. Both these codes enable simplified ML-decoding (see Section \ref{sysmod} for a definition of simplified ML-decoding), achieving a complexity of the order of $M^3$ in general, and $M^2$ for square QAM (shown in Section \ref{mlexist}). These codes have a slightly lower coding gain than the Golden code and hence show a slight loss in performance compared to the Golden code. These codes sacrifice the coding gain for simplified ML-decoding complexity.

For 4 transmit antennas, the popular codes are the quasi-orthogonal designs, first introduced in  \cite{JH} and the CIOD for 4 transmit antenna \cite{ZS}, both of which are rate one codes. The CIOD is known to be single symbol decodable and the MDC-QOD \cite{chau} is also single symbol decodable. But when 2 or more receive antennas are employed, these codes cannot be considered to be full-rate. The perfect code for 4 transmit antennas has a high rate of 4 complex symbols per channel use but its use in practice is hampered by its high decoding complexity, even with the use of sphere decoding \cite{ViB}, \cite{DCB}. For a $4 \times 2$ MIMO system, the best performing code has been the DjABBA code \cite{HTW}, which beats even the punctured perfect code for 4 transmit antennas in performance \cite{BHV,HHVMM}. This code was designed for performance alone and  has a high ML-decoding complexity, of the order of $M^7$, as shown in Section \ref{mlexist}. The first attempt at reducing the ML-decoding complexity for a $4\times2$ system while maintaining full-rate was made by Biglieri, Hong and Viterbo \cite{BHV}. The full-rate code that they have proposed, which we call the BHV code, has an ML-decoding complexity of the order of $M^6$ for general constellations, (though this has been reported to be $M^7$ in \cite{BHV}), but doesn't have full-diversity. However, the code matches the DjABBA code in the low SNR scenario and betters the punctured perfect code in codeword error performance (CER).

The contributions of this paper are as follows
 \begin{itemize}
 \item we propose a new full-rate, full-diversity STBC for the $2\times2$ MIMO system. This code has an ML-decoding complexity of the order of $M^3$ in general, as compared to $M^4$ for the Golden code. For square QAM, the ML-decoding complexity of our code is of the order of $M^2\sqrt{2}$, the same as that of the Golden code.
 \item Our code also matches the Golden code in coding gain for QAM constellations and is shown to have the non-vanishing determinant (NVD) property for QAM constellations and hence, is DMG optimal. We also show that our code is \emph {information-lossless}.
 \item We propose a new full-rate, full-diversity STBC for $4\times2$ MIMO systems, having ML-decoding complexity of the order of $M^5$ for arbitrary complex constellations, and of the order of $M^4\sqrt{M}$ for square QAM constellations, whereas the corresponding complexity for the DjABBA code are $M^7$ and $M^6$ respectively. It also has a higher coding gain than the DjABBA code for 4- and 16-QAM constellations and hence, a better CER performance.
 \item We state the conditions that allow simplified ML-decoding and show that for square QAM constellations, the ML-decoding complexity of the Golden code can be reduced to $ M^2 \sqrt{M}.$
 \end{itemize}

The remaining content of the paper is organized as follows : In Section \ref{sysmod}, we give the system model and the code design criteria. In Section \ref{2by2}, we present our code for the $2\times2$ MIMO system and show that it is information-lossless. In Section \ref{nvd}, we show that our code has the NVD property and DMG optimality. In Section \ref{4by2}, we present our code for the $4\times2$ MIMO system. Section \ref{lowdec} deals with the low complexity ML-decoding of these codes. In Section \ref{mlexist}, we analyze the ML-decoding complexity for the Golden code, the HTW-PGA code, the DjABBA code and the BHV code. The simulations results  constitute Section \ref{dis}. Concluding remarks are made in Section \ref{conc}.

\textit{\textbf{Notations}:} Throughout, bold, lowercase letters are used to denote vectors and bold, uppercase letters are used to denote matrices. Let $\textbf{X}$ be a complex matrix. Then $\textbf{X}^T$, $\textbf{X}^{H}$ and $det\left[\textbf{X}\right]$ denote the transpose, Hermitian and determinant of $\textbf{X}$, respectively. For a complex variable $x,$ $x_I$ and $x_Q$ denote the real and imaginary part of $x,$ respectively. Also, $j$ represents $\sqrt{-1}$ and the sets of all integers, all real and complex numbers are denoted by $\mathbb{Z},$  $\mathbb{R}$ and $\mathbb{C},$ respectively. The Frobenius norm and the trace operations are denoted by $\Vert.\Vert_F$ and $tr\left[.\right] $ respectively. The operation of stacking the columns of $\textbf{X}$ one below the other is denoted by $vec(\textbf{X}).$ The Kronecker product is denoted by $\otimes$, $\textbf{I}_{T}$ and $\textbf{O}_T$ denote the $T\times T$ identity matrix and the null matrix, respectively. The inner product of two vectors $\textbf{x}$ and $\textbf{y}$ is denoted by $\langle \textbf{x},\textbf{y} \rangle$.  For a complex random variable $X$, $X \sim \mathcal{N}_{\mathbb{C}}(0,N)$ denotes that $X$ has a complex normal distribution with mean $0$ and variance $N$. For any real number $m$, \emph{rnd}[$m$] denotes the operation that rounds off $m$ to the nearest integer, i.e.,
\begin{equation*}
\emph{rnd}[m] = \left\{ \begin{array}{rr}
\lfloor m \rfloor  & \textrm{if} ~ \lceil m \rceil - m > m - \lfloor m \rfloor\\
\lceil m \rceil  & \textrm{otherwise}\\
\end{array} \right.
\end{equation*}
For a complex variable $x$, the $\check{(.)}$ operator acting on $x$ is defined as follows
\begin{equation*}
\check{x} \triangleq \left[ \begin{array}{rr}
                             x_I & -x_Q \\
                             x_Q & x_I \\
                            \end{array}\right]
\end{equation*}
The $\check{(.)}$ can similarly be applied to any matrix $\textbf{X} \in \mathbb{C}^{n \times m}$ by replacing each entry $x_{ij}$ by $\check{x}_{ij}$, $i=1,2,\cdots, n, j = 1,2,\cdots,m$ , resulting in a matrix denoted by $\check{\textbf{X}} \in \mathbb{R}^{2n \times 2m}$.

Given a complex vector $\textbf{x} = [ x_1, x_2, \cdots, x_n ]^T$, $\tilde{\textbf{x}}$ is defined as
\begin{equation*}
\tilde{\textbf{x}} \triangleq [ x_{1I},x_{1Q}, \cdots, x_{nI}, x_{nQ} ]^T.
\end{equation*}
\noindent and
$\tilde{\textbf{x}}^\prime$ is defined as
\begin{equation*}
\tilde{\textbf{x}}^\prime \triangleq [ -x_{1Q},x_{1I}, , \cdots, -x_{nQ},x_{nI}]^T.
\end{equation*}
It follows that $\check{\textbf{x}} = [\tilde{\textbf{x}} \ \tilde{\textbf{x}}^\prime]$.

\section{System Model}
\label{sysmod}
We consider Rayleigh quasi-static flat-fading MIMO channel with full channel state information (CSI) at the receiver but not at the transmitter. For $n_t\times n_r$ MIMO transmission, we have
\begin{equation}\label{Y}
\textbf{Y = HS + N}
\end{equation}

\noindent where $\textbf{S} \in \mathbb{C}^{n_t\times T}$ is the codeword matrix, transmitted over T channel uses, $\textbf{N} \in \mathbb{C}^{n_r \times T}$ is a complex white Gaussian noise matrix with i.i.d entries $\sim
\mathcal{N}_{\mathbb{C}}\left(0,N_{0}\right)$ and $\textbf{H} \in \mathbb{C}^{n_r\times n_t}$ is the channel matrix with the entries assumed to be i.i.d circularly symmetric Gaussian random variables $\sim \mathcal{N}_\mathbb{C}\left(0,1\right)$. $\textbf{Y} \in \mathbb{C}^{n_r \times T}$ is the received matrix.

\begin{definition}\label{def1}$\left(\textbf{Code rate}\right)$ If there are $k$ independent complex information symbols in the codeword which are transmitted over $T$ channel uses, then, the code rate is defined to be $k/T$  complex symbols per channel use. For instance, for the Alamouti code, $k=2$ and $T=2.$ So, its code rate is 1 complex symbol per channel use.
\end{definition}

\begin{definition}\label{def2}$\textbf{(Full-rate code)}$. An STBC is said to be \emph{full-rate} if it transmits at the rate of $n_{min}$ complex symbols per channel use, where $n_{min} = min\left(n_t,n_r\right)$.
\end{definition}

So, the Alamouti code can be considered to be full-rate for $2\times1$ MIMO alone, while the Golden code is full-rate for $n_r \geq 2$.

Considering ML-decoding, the decoding metric that is to be minimized over all possible values of codewords $\textbf{S}$ is given by
 \begin{equation}
\label{ML}
 \textbf{M}\left(\textbf{S}\right) = \Vert \textbf{Y} - \textbf{HS} \Vert_F^2
 \end{equation}

\begin{definition}\label{def3}$\left(\textbf{Decoding complexity}\right)$
The ML decoding complexity is a  measure of the maximum number of symbols that need to be jointly decoded in minimizing the ML decoding metric. This number can be $k$ in the worst scenario, $k$ being the total number of information symbols in the code. Such a code is said to have a high ML-decoding complexity, of the order of $M^k$, where $M$ is the size of the signal constellation. If the code has an ML-decoding complexity of order less than $M^k$, the code is said to admit \emph{\textbf{simplified ML-decoding}}. For some codes, all the symbols can be independently decoded. Such codes are said to be \emph{single-symbol decodable}.
\end{definition}
\begin{definition}\label{def4}$\left(\textbf{Generator matrix}\right)$ For any STBC $\textbf{S}$ that encodes $k$ information symbols, the \textbf{\emph{generator}} matrix $\textbf{G}$ is defined by the following equation  \cite{BHV}
\begin{equation}
\widetilde{vec\left(\textbf{S}\right)} = \textbf{G} \tilde{\textbf{s}}.
\end{equation}
\noindent where $\textbf{s} \triangleq \left[ s_1, s_2,\cdots,s_k \right]^T$ is the information symbol vector.
\end{definition}

An STBC can be expressed in terms of its \textbf{\emph{weight matrices}} (linear dispersion matrices) as follows
\begin{equation}
\textbf{S} = \sum_{i=1}^{k}\textbf{A}_{2i-1}s_{iI}+\textbf{A}_{2i}s_{iQ}
\end{equation}
Here, $\textbf{A}_i,i=1,2,\cdots,2k$ are the weight matrices for $\textbf{S}$. It follows that
\begin{equation}
\textbf{G} = [\widetilde{vec(\textbf{A}_1)}\ \widetilde{vec(\textbf{A}_2)}\ \cdots \ \widetilde{vec(\textbf{A}_{2k})}]
\end{equation}

It is well known  \cite{TSC}, that an analysis of the PEP leads to the following design criteria:
\begin{enumerate}
\item \emph{\textbf{Rank criterion}}: To achieve maximum diversity, the codeword difference matrix $(\textbf{S} - \hat{\textbf{S}})$ must have full-rank for all possible codeword pairs and the diversity gain is $n_tn_r$. If full-rank is not achievable, then, the diversity gain is given by $rn_r$, where $r$ is the minimum rank of the codeword difference matrix over all possible codeword pairs.

\item \emph{\textbf{Determinant criterion }}: For a full ranked STBC, the minimum determinant $\delta_{min}$, defined as
\begin{equation}
 \delta_{min} \triangleq \min_{\textbf{S} \neq \hat{\textbf{S}}} det\left[\left(\textbf{S}-\hat{\textbf{S}}\right)\left(\textbf{S}-\hat{\textbf{S}}\right)^{H}\right]
\end{equation}
should be maximized. The coding gain is given by $\left(\delta_{min}\right)^{1/n_t}$, with $n_t$ being the number of transmit antennas.
\end{enumerate}

If the STBC is non full-diversity and $r$ is the minimum rank of the codeword difference matrix over all possible codeword pairs, then, the coding gain $\delta$ is given by
\begin{equation*}
\delta = \min_{\textbf{S} -\hat{\textbf{S}}}\left(\prod_{i=1}^{r}\lambda_i\right)^{\frac{1}{r}}
\end{equation*}
\noindent where  $\lambda_i, i = 1,2,\cdots,r$, are the non-zero eigen values of the matrix $\left(\textbf{S} - \hat{\textbf{S}}\right)\left(\textbf{S} - \hat{\textbf{S}}\right)^H$.
It should be noted that for high signal-to-noise ratio (SNR) values at each receive antenna, the dominant parameter is the diversity gain which defines the slope  of the CER curve. This implies that it is important to first ensure full-diversity of the STBC and then try to maximize the coding gain.

\section{ The Proposed STBC For $2\times2$ MIMO And Information-Losslessness}
\label{2by2}
In this section, we present our STBC \cite{Pav1}, \cite{Pav2} for  $2\times 2$ MIMO system. The design is based on the CIODs, which were studied in \cite{ZS} in connection with a general class of single-symbol decodable codes which includes complex orthogonal designs as a proper subclass. Specifically, for 2 transmit antennas, the CIOD is as follows.
\begin{definition}
The CIOD for $2$ transmit antennas \cite{ZS} is \\
\begin{equation}
\label{ciod}
\textbf{X}(s_1,s_2) = \left[\begin{array}{cc}
            s_{1I}+js_{2Q} & 0 \\
            0  & s_{2I}+js_{1Q}\\
           \end{array}\right]
\end{equation}
where $s_i \in \mathbb{C}, i = 1,2$ are the information symbols and $s_{iI}$ and $s_{iQ}$ are the in-phase (real) and quadrature-phase (imaginary) components of $s_i,$ respectively. Notice that in order to make the above STBC full-rank, the signal constellation $\mathcal{A}$ from which the symbols $s_i$ are chosen should be such that the real part (imaginary part, resp.) of any signal point in $\mathcal{A}$ is not equal to the real part (imaginary part, resp.) of any other signal point in $\mathcal{A}$ \cite{ZS}. So if QAM constellations are chosen, they have to be rotated. The optimum angle of rotation has been found in \cite{ZS} to be $\frac{1}{2} tan^{-1}2$ radians and this maximizes the diversity and coding gain. We denote this angle by $\theta_g.$
\end{definition}

The proposed  $2 \times 2$ STBC $\mathbf{S}$ is given by
\begin{equation}\label{2code}
 \textbf{S}(x_1,x_2,x_3,x_4) = \textbf{X}\left(s_1,s_2\right) + e^{j\theta}\textbf{X}\left(s_3,s_4\right)\textbf{P}
\end{equation}
where
\begin{itemize}
\item The four symbols $s_1,s_2,s_3$ and $s_4 \in \mathcal{A}$, where $\mathcal{A}$ is a $\theta_g$ radians rotated version of an integer QAM signal set, denoted by $\mathcal{A}_q$, which is a finite subset of the integer lattice, and $x_1,x_2,x_3,x_4 \in \mathcal{A}_q$, i.e, $s_i=e^{j\theta_g}x_i,~~~i=1,2,3,4.$

\item $\textbf{P}$ is a permutation matrix designed to make the STBC full-rate and is given by $
 \textbf{P} = \left[\begin{array}{cc}
0 & 1 \\
1 & 0 \\
\end{array}\right].$
\item The choice of $\theta$ in the above expression should be such that the diversity and coding gain are maximized. We choose $\theta$ to be $\pi/4$ and show in the next section that this angle maximizes the coding gain.
\end{itemize}

Explicitly, our code matrix is

\begin{equation}
\label{2by2code}
 \textbf{S}(x_1,x_2,x_3,x_4)  =  \left[\begin{array}{rr}
 s_{1I}+js_{2Q} & e^{j\pi/4}(s_{3I}+js_{4Q})  \\
 e^{j\pi/4}(s_{4I}+js_{3Q}) & s_{2I}+js_{1Q}
\end{array}\right] \\
\end{equation}
\noindent with  $s_{iI} = cos (\theta_g) x_{iI} - sin (\theta_g)  x_{iQ}$ and $s_{iQ} = sin (\theta_g) x_{iI} + cos (\theta_g)  x_{iQ}$, $i = 1,2,3,4$.

The minimum determinant for our code when the symbols are chosen from the regular QAM constellations (one in which the difference between any two signal points is a multiple of 2) is $3.2$, the same as that for the Golden code, which is proved in the next section.
The generator matrix for our STBC (as defined in Definition \ref{def4}), corresponding to the information vector consisting of symbols $x_i$, is as follows:
\begin{equation}\label{gen}
G = \left[\begin{array}{cccccccc}
      cos(\theta_g) & -sin(\theta_g) & 0 & 0 & 0 & 0 & 0 & 0 \\
      0 & 0 & sin(\theta_g) & cos(\theta_g) & 0 & 0 & 0 & 0 \\
      0 & 0 & 0 & 0 & -\frac{sin(\theta_g)}{\sqrt{2}} & -\frac{cos(\theta_g)}{\sqrt{2}} & \frac{cos(\theta_g)}{\sqrt{2}} & -\frac{sin(\theta_g)}{\sqrt{2}}\\
      0 & 0 & 0 & 0 & \frac{sin(\theta_g)}{\sqrt{2}} & \frac{cos(\theta_g)}{\sqrt{2}} & \frac{cos(\theta_g)}{\sqrt{2}} & -\frac{sin(\theta_g)}{\sqrt{2}}\\
      0 & 0 & 0 & 0 & \frac{cos(\theta_g)}{\sqrt{2}} & -\frac{sin(\theta_g)}{\sqrt{2}} & -\frac{sin(\theta_g)}{\sqrt{2}} & -\frac{cos(\theta_g)}{\sqrt{2}}\\
      0 & 0 & 0 & 0 & \frac{cos(\theta_g)}{\sqrt{2}} & -\frac{sin(\theta_g)}{\sqrt{2}} & \frac{sin(\theta_g)}{\sqrt{2}} & \frac{cos(\theta_g)}{\sqrt{2}}\\
      0 & 0 & cos(\theta_g) & -sin(\theta_g) & 0 & 0 & 0 & 0\\
      sin(\theta_g) & cos(\theta_g) & 0 & 0 & 0 & 0 & 0 & 0 \\
      \end{array}\right]
\end{equation}
It is easy to see that this generator matrix is orthonormal. In \cite{JJK}, it was shown that a sufficient condition for an STBC to be \emph{\textbf{information-lossless}} is that its generator matrix should be unitary. Hence, our STBC has the \emph{\textbf{information-losslessness}} property.

\section{NVD Property And DMG Optimality Of The $2\times2$ Code}
\label{nvd}
In this section we show that the proposed code has the NVD property \cite{BRV}, which, in conjunction with full-rateness, means that our code is  DMG tradeoff optimal \cite{PVK}. We also show that the angle $\pi/4$ in \eqref{2code} maximizes the coding gain.

\begin{theorem}\label{nvdproof}
The minimum determinant of the proposed $2\times2$ code, given by \eqref{2by2code},  when the symbols are chosen from $\mathbb{Z}[j]$ is $1/5$.
\end{theorem}

\begin{proof}
 The determinant of the codeword matrix $\textbf{S}$ can be written as
\begin{equation}\label{equ1}
det(\textbf{S}) = (s_{1I}+js_{2Q})(s_{2I}+js_{1Q}) - j[(s_{3I}+js_{4Q})(s_{4I}+js_{3Q})].
\end{equation}
Using $s_{iI} = (s_i+s_i^*)/2$ and $js_{iQ} = (s_i-s_i^*)/2$ in equation \eqref{equ1}, we get,
{\footnotesize
\begin{eqnarray*}
4det(\textbf{S}) & = &(s_1+s_1^*+s_2-s_2^*)(s_2+s_2^*+s_1-s_1^*) -j[(s_3+s_3^*+s_4-s_4^*)(s_4+s_4^*+s_3-s_3^*)]\\
 & = &\big((s_1+s_2)+(s_1-s_2)^*\big)\big((s_1+s_2)-(s_1-s_2)^*\big)-j[\big((s_3+s_4)+(s_3-s_4)^*\big)\big((s_3+s_4)-(s_3-s_4)^*\big)].
\end{eqnarray*}
}
Since $s_i = e^{j\theta_g}x_i, i = 1,2,3,4$, with $s_i \in \mathcal{A}$, $x_i \in \mathcal{A}_q$, a subset of $\mathbb{Z}[j]$,  defining $A \triangleq (x_1+x_2)$, $B \triangleq (x_1-x_2)^*$, $C \triangleq (x_3+x_4)$ and $D \triangleq (x_3-x_4)^*$, with $A,B,C$ and $D \in \mathbb{Z}[j]$, we get
\begin{eqnarray*}
4det(\textbf{S})& = & (e^{j\theta_g}A+e^{-j\theta_g}B)(e^{j\theta_g}A-e^{-j\theta_g}B)-j[(e^{j\theta_g}C+e^{-j\theta_g}D)(e^{j\theta_g}C-e^{-j\theta_g}D)]\\
& = & e^{j2\theta_g}A^2-e^{-j2\theta_g}B^2 - j(e^{j2\theta_g}C^2-e^{-j2\theta_g}D^2).
\end{eqnarray*}
Since $e^{j2\theta_g} = cos(2\theta_g)+jsin(2\theta_g)=(1+2j)/\sqrt{5}$, we get
\begin{equation}\label{mod}
4\sqrt{5}det(\textbf{S}) =  (1+2j)(A^2-jC^2) - (1-2j)(B^2-jD^2).
\end{equation}
For the determinant of $\textbf{S}$ to be 0, we must have
\begin{eqnarray*}
 (1+2j)(A^2-jC^2)& = &(1-2j)(B^2-jD^2) \\
\Rightarrow (1+2j)^2(A^2-jC^2)& = &5(B^2-jD^2).
\end{eqnarray*}
The above can be written as
\begin{equation}\label{ABC}
 A_1^2-jC_1^2 = 5(B^2-jD^2)
\end{equation}
where $A_1 = (1+2j)A,C_1 = (1+2j)C$ and clearly $A_1,C_1 \in \mathbb{Z}[j]$. It has been shown in \cite{DV} that  \eqref{ABC} holds only when $A_1 = B = C_1 = D = 0$, i.e., only when $x_1 = x_2 =x_3 = x_4 =0$. This means that the determinant of the codeword difference matrix is 0 only when the codeword difference matrix is itself the zero matrix. So, for any distinct pair of codewords, the codeword difference matrix is always full-rank for any constellation which is a subset of $\mathbb{Z}[j]$. Also, the minimum value of the modulus of the R.H.S of \eqref{mod} can be seen to be $4$. This occurs for $(A,B,C,D) = (1,1,0,0)$ or $(0,0,1,1)$. The occurrence of any other combination of $A,B,C$ and $D$ that results in a lower value of the modulus of the R.H.S of \eqref{mod} can be ruled out after noting that $x_1,x_2,x_3$ and $x_4$ take only values from $\mathbb{Z}[j]$. For eg. $(A,B,C,D) = (1,j,0,0)$ is one such combination, but it is easy to see mathematically that such a combination cannot occur for $x_i \in \mathbb{Z}[j],i=1,2,3,4$. So, $\vert det(\textbf{S}) \vert \geq 1/\sqrt{5}$, meaning that the minimum determinant for the code is $1/5$.
\end{proof}

In particular, when the constellation chosen is the regular QAM constellation, the difference between any two signal points is a multiple of 2. Hence, for such constellations, $\vert det(\textbf{S-S}^\prime) \vert \geq 4/\sqrt{5}$, where $\textbf{S}$ and $\textbf{S}^\prime$ are distinct codewords. The minimum determinant is consequently 16/5 and hence the proposed code has the NVD property \cite{BRV}. Now, from \cite{PVK}, where it was shown that full-rate codes which satisfy the NVD property achieve the optimal DMG tradeoff, our proposed STBC is DMG tradeoff optimal.

As a byproduct of Theorem \ref{nvdproof}, we arrive at the following lemma.
\begin{lemma}
The choice of $\pi/4$ for $\theta$ in \eqref{2code} maximizes the coding gain of the proposed $2\times2$ code for QAM constellations.
\end{lemma}

\begin{proof}
Consider the CIOD whose codeword has the structure shown in \eqref{ciod}. The set of codeword difference matrices of the CIOD is a subset of the set of the codeword difference matrices of the proposed $2\times2$ code, whose codeword structure is given in \eqref{2by2code}. It is to be noted that the minimum determinant and hence the coding gain of a code depend on the codeword difference matrices of the code. In \eqref{mod}, if we let $C=D=0$, we arrive at the expression for the determinant of a codeword matrix of the CIOD. So, for the CIOD, whose codeword matrix is denoted by $\textbf{S}$, we have
\begin{equation}\label{mod1}
4\sqrt{5}det(\textbf{S}) =  (1+2j)A^2 - (1-2j)B^2.
\end{equation}
where, $A = (x_1+x_2)$ and $B = (x_1-x_2)^*$, with $x_1$ and $x_2$ taking values from $\mathbb{Z}[j]$. It is evident that the minimum of the modulus of the R.H.S of \eqref{mod1} is $4$, which occurs for $A=B=1$. So, the minimum of the absolute value of the determinant of a codeword matrix of the CIOD when the symbols take values from $\mathbb{Z}[j]$ (not all taking zero values) is $1/\sqrt{5}$. When the symbols take values from the regular QAM constellation, the minimum of the absolute value of determinant of a non-zero codeword difference matrix is $4/\sqrt{5}$ and hence, the minimum determinant for the CIOD is $16/5$. We have already shown that the minimum determinant for our $2\times$ code is $16/5$, when the symbols take values from the regular QAM. This shows that the choice of $\pi/4$ for $\theta$ in \eqref{2code} indeed maximizes the coding gain.
\end{proof}

\section{ The Proposed STBC For The $4\times2$ MIMO System}
\label{4by2}
In this section,  we present our STBC for the $4\times 2$ MIMO system \cite{Pav3} following the same approach that we took to design the $2\times2$ code. The design is based on the CIOD for 4 antennas, whose structure is as defined below.
\begin{definition}
 CIOD for $4$ transmit antennas \cite{ZS} is as follows:\\
\begin{equation}
\label{ciod4}
\textbf{X}(s_1,s_2,s_3,s_4) = \left[\begin{array}{cccc}
            s_{1I}+js_{3Q} & -s_{2I}+js_{4Q} & 0 & 0\\
            s_{2I}+js_{4Q} & s_{1I}-js_{3Q} & 0 & 0\\
            0 & 0 & s_{3I}+js_{1Q} & -s_{4I}+js_{2Q}\\
            0 & 0 & s_{4I}+js_{2Q} & s_{23}-js_{1Q}\\
           \end{array}\right]
\end{equation}
where $s_i, i = 1,\cdots,4$ are the information symbols as defined in the previous section. Here again, the symbols are chosen from a rotated version of the regular QAM constellation, with $\theta_g$ being the angle of rotation.
\end{definition}

The proposed STBC is obtained as follows. Our $4\times4$ code matrix, denoted by $\textbf{S}$, encodes eight symbols $x_1,\cdots,x_8$ drawn from a QAM constellation, denoted by $\mathcal{A}_q$. As before, we denote the rotated version of $\mathcal{A}_q$ by $\mathcal{A}$. Let $s_i  \triangleq e^{j\theta_g}x_i, i = 1,2, \cdots 8$, so that the symbols $s_i$ are drawn from the constellation $\mathcal{A}$. The codeword matrix is defined as
\begin{equation}\label{4by2code}
 \textbf{S}(x_1,x_2, \cdots, x_8) \triangleq \textbf{X}(s_1,s_2,s_3,s_4) + e^{j\theta}\textbf{X}(s_5,s_6,s_7,s_8)\textbf{P}
\end{equation}
\noindent with $\theta \in [0,\pi/2]$ and $\textbf{P}$ being a permutation matrix designed to make the STBC full-rate and given by
\begin{equation*}
 \textbf{P} = \left[\begin{array}{cccc}
0 & 0 & 1 & 0\\
0 & 0 & 0 & 1\\
1 & 0 & 0 & 0\\
0 & 1 & 0 & 0\\
\end{array}\right].
\end{equation*}
The choice of $\theta$ is to maximize the diversity and coding gain. Here again, we take $\theta$ to be $\pi/4$. This value of $\theta$ provides the largest coding gain achievable for this family of codes. This is so because the minimum determinant for the CIOD as defined in \eqref{ciod4} (which can also be obtained by letting the variables $s_5$, $s_6$, $s_7$ and $s_8$ be zeros in \eqref{4by2code})  is 10.24 \cite{chau} for unnormalized QAM constellations. The value of the minimum determinant for our $4\times2$ code, obtained for unnormalized 4-QAM and 16-QAM constellations is 10.24, which was checked by exhaustive search. This shows that the choice of $\pi/4$ maximizes the  coding gain. The resulting code matrix is as shown below.
\begin{eqnarray*}
 \textbf{S} = \left[\begin{array}{cccc}
            s_{1I}+js_{3Q} & -s_{2I}+js_{4Q} & e^{j\pi/4}(s_{5I}+js_{7Q}) & e^{j\pi/4}(-s_{6I}+js_{8Q})\\
            s_{2I}+js_{4Q} & s_{1I}-js_{3Q} & e^{j\pi/4}(s_{6I}+js_{8Q}) & e^{j\pi/4}(s_{5I}-js_{7Q})\\
            e^{j\pi/4}(s_{7I}+js_{5Q}) & e^{j\pi/4}(-s_{8I}+js_{6Q}) & s_{3I}+js_{1Q} & -s_{4I}+js_{2Q}\\
            e^{j\pi/4}(s_{8I}+js_{6Q}) & e^{j\pi/4}(s_{7I}-js_{5Q}) & s_{4I}+js_{2Q} & s_{3I}-js_{1Q}\\
           \end{array}\right]
\end{eqnarray*}

This code is full-rate only for the $4\times2$ MIMO system, unlike the perfect space time code \cite{ORBV}, which is full-rate for $n_r \geq 4$. Also, the generator matrix for our code can be checked to be non-unitary. So, our STBC for $4\times2$ MIMO system is not information-lossless.

\section{ Low Complexity ML-Decoding Of The $2\times2$ And $4\times2$ Codes }
\label{lowdec}
In this section, we show how our codes admit simplified ML-decoding. The information symbols are assumed to take values from QAM constellations. In the general setting, it can be shown that \eqref{Y} can be written as
\begin{equation}
 \widetilde{vec(\textbf{Y})} = \textbf{H}_{eq}\tilde{\textbf{x}} + \widetilde{vec(\textbf{N})}
\end{equation}
\noindent where $\textbf{H}_{eq} \in \mathbb{R}^{2n_rT\times2k}$ is given by
\begin{equation}
 \textbf{H}_{eq} = \left(\textbf{I}_T \otimes \check{\textbf{H}}\right)\textbf{G}
\end{equation}
with $\textbf{G} \in \mathbb{R}^{2n_tT\times 2k}$ being the generator matrix as in Definition \ref{def4}, so that $\widetilde{vec\left(\textbf{S}\right)} = \textbf{G} \tilde{\textbf{x}}.$ and
\begin{equation*}
\tilde{\textbf{x}} \triangleq [x_{1I},x_{1Q},\cdots,x_{kI},x_{kQ}]^T
\end{equation*}
\noindent with $x_i,i=1,\cdots,k$ drawn from $\mathcal{A}_q$, which is the regular QAM constellation. Using this equivalent model, the ML decoding metric can be written as
\begin{equation}
 \textbf{M}\left(\tilde{\textbf{x}}\right) = \Vert \widetilde{vec\left(\textbf{Y}\right)} - \textbf{H}_{eq}\tilde{\textbf{x}}\Vert^2
\end{equation}
On obtaining the \textbf{QR} decomposition of $\textbf{H}_{eq}$, we get $\textbf{H}_{eq} $ = $\textbf{QR}$, where  $\textbf{Q} \in \mathbb{R}^{2n_rT\times 2k}$ is an orthonormal matrix and  $\textbf{R} \in \mathbb{R}^{2k \times 2k}$ is an upper triangular matrix. The ML decoding metric now can be written as
\begin{equation}
 \textbf{M}(\tilde{\textbf{x}}) = \Vert \textbf{Q}^T\widetilde{vec\textbf{(Y)}} - \textbf{R}\tilde{\textbf{x}}\Vert^2 = \Vert \textbf{y}^\prime - \textbf{R}\tilde{\textbf{x}}\Vert^2
\end{equation}
\noindent where $\textbf{y}^\prime \triangleq [ y_1^\prime, \cdots, y_{2k}^\prime ] = \textbf{Q}^T\widetilde{vec\textbf{(Y)}}$. If $\textbf{H}_{eq} \triangleq [ \textbf{h}_1 \ \textbf{h}_2 \cdots \textbf{h}_{2k} ]$, where $\textbf{h}_i, i = 1,2,\cdots,2k$ are column vectors, then $\textbf{Q}$ and $\textbf{R}$ have the general form obtained by $Gram-Schmidt$ process as shown below
\begin{equation*}
 \textbf{Q} \triangleq [ \textbf{q}_1\ \textbf{q}_2 \ \textbf{q}_3 \cdots \textbf{q}_{2k} ]
\end{equation*}
where $\textbf{q}_i, i = 1,2,\cdots,2k$ are column vectors, and
\begin{equation*}
 \textbf{R} \triangleq \left[\begin{array}{ccccc}
\Vert \textbf{r}_1 \Vert & \langle \textbf{q}_1,\textbf{h}_2 \rangle & \langle \textbf{q}_1,\textbf{h}_3 \rangle & \ldots &  \langle \textbf{q}_1,\textbf{h}_{2k} \rangle\\
0 & \Vert \textbf{r}_2 \Vert & \langle \textbf{q}_2,\textbf{h}_3 \rangle & \ldots & \langle \textbf{q}_2,\textbf{h}_{2k} \rangle\\
0 & 0 &  \Vert \textbf{r}_3 \Vert & \ldots & \langle \textbf{q}_3,\textbf{h}_{2k} \rangle\\
\vdots & \vdots & \vdots & \ddots & \vdots\\
0 & 0 & 0 & \ldots & \Vert \textbf{r}_{2k} \Vert\\
\end{array}\right]
\end{equation*}
\noindent where $\textbf{r}_1 = \textbf{h}_1$, $\textbf{q}_1 = \frac{\textbf{r}_1}{\Vert \textbf{r}_1 \Vert}$, $\textbf{r}_i = \textbf{h}_i - \sum_{j=1}^{i-1}\langle \textbf{q}_j,\textbf{h}_i \rangle \textbf{q}_j,$ $\ \textbf{q}_i = \frac{\textbf{r}_i}{\Vert \textbf{r}_i \Vert},\ i = 2,3,\cdots,2k$.

\begin{lemma}\label{lemma1}  Let $\textbf{M} = [ \textbf{f}_1 \ \textbf{f}_2 \ \cdots \ \textbf{f}_n ][ \textbf{g}_1 \ \textbf{g}_2  \ \cdots \ \textbf{g}_n ]^T$, where $\textbf{f}_i \triangleq [f_{i1}, f_{i2},\cdots,f_{in}]$, $\textbf{g}_i \triangleq [g_{i1}, g_{i2},\cdots,g_{in}]$ $\in \mathbb{R}^{n \times 1}, i = 1,2, \cdots, n$. Then,
$tr(\textbf{M}) =  \sum_{i=1}^{n} \langle \textbf{g}_i, \textbf{f}_i \rangle$.
\end{lemma}

\begin{proof}
From the definition of the trace operation, we have
\begin{eqnarray*}
tr(\textbf{M})&  =  & \sum_{j=1}^{n} \sum_{i=1}^{n} f_{ij}g_{ij} \\
              &  =  & \sum_{i=1}^{n} \sum_{j=1}^{n} g_{ij}f_{ij} = \sum_{i=1}^{n} \langle \textbf{g}_i, \textbf{f}_i \rangle.
\end{eqnarray*}
\vspace{-0.5cm}
\end{proof}

\begin{theorem}\label{thm1} For an STBC with $k$ independent complex symbols and $2k$ weight matrices $\textbf{A}_l, l=1,2,\cdots,2k$, if, for any $i$ and $j$, $i \neq j,1 \leq i,j \leq 2k$, $\textbf{A}_i\textbf{A}_j^H + \textbf{A}_j\textbf{A}_i^H = \textbf{O}_{n_t}$, then, the $i^{th}$ and the $j^{th}$ columns of the equivalent channel matrix $\textbf{H}_{eq}$ are orthogonal.
\end{theorem}
\begin{proof}
We note that the following identities hold for matrices $\textbf{A} \in \mathbb{C}^{m \times n},\textbf{B} \in \mathbb{C}^{m \times p}, \textbf{C} \in \mathbb{C}^{p \times n}$ and vectors $\textbf{x} \in \mathbb{C}^{p \times 1}, \textbf{z} \in \mathbb{C}^{p \times 1}$.
\vspace{-1.5 cm}
\begin{center}
\begin{equation}\label{rel1}
\textbf{A} = \textbf{B}\textbf{C}  \Leftrightarrow \check{\textbf{A}} = \check{\textbf{B}}\check{\textbf{C}}
\end{equation}
\begin{equation}\label{cda}
\langle \tilde{\textbf{z}},\tilde{\textbf{x}} \rangle  = \langle \tilde{\textbf{z}}^\prime,\tilde{\textbf{x}}^\prime \rangle
\end{equation}
\end{center}
With these identities, we proceed as follows
\begin{center}
$\textbf{A}_i\textbf{A}_j^H + \textbf{A}_j\textbf{A}_i^H = \textbf{O}_{n_t} \Leftrightarrow \textbf{H}\textbf{A}_i\textbf{A}_j^\textbf{H}\textbf{H}^\textbf{H} + \textbf{H}\textbf{A}_j\textbf{A}_i^\textbf{H}\textbf{H}^\textbf{H} = \textbf{O}_{n_r}$.
\end{center}
Applying the $\check{(.)}$ operator and using \eqref{rel1}, we get
\begin{equation}
\check{\textbf{H}}\check{\textbf{A}}_i(\check{\textbf{H}}\check{\textbf{A}}_j)^T + \check{\textbf{H}}\check{\textbf{A}}_j(\check{\textbf{H}}\check{\textbf{A}}_i)^T = \textbf{O}_{2n_r}
\end{equation}
This indicates that the real matrix $\textbf{M} \triangleq \check{\textbf{H}}\check{\textbf{A}}_i(\check{\textbf{H}}\check{\textbf{A}}_j)^T$ is a skew-symmetric matrix and hence its diagonal elements are zeros. Let $\textbf{A}_i \triangleq [\textbf{a}_{i,1} \ \textbf{a}_{i,2} \ \cdots \ \textbf{a}_{i,T}]$, where $\textbf{a}_{i,k}, k = 1,2,\cdots,T$ are the columns of $\textbf{A}_i$. Then, $\check{\textbf{A}}_i = [ \tilde{\textbf{a}}_{i,1} \ \tilde{\textbf{a}}_{i,1}^\prime \ \cdots \ \tilde{\textbf{a}}_{i,T} \ \tilde{\textbf{a}}_{i,T}^\prime ]$. Therefore,
\begin{equation*}
\textbf{M} = [ \check{\textbf{H}}\tilde{\textbf{a}}_{i,1} \ \check{\textbf{H}}\tilde{\textbf{a}}_{i,1}^{\prime} \ \cdots \ \check{\textbf{H}}\tilde{\textbf{a}}_{i,T} \ \check{\textbf{H}}{\tilde{\textbf{a}}_{i,T}}^{\prime} ][ \check{\textbf{H}}\tilde{\textbf{a}}_{j,1} \ \check{\textbf{H}}\tilde{\textbf{a}}_{j,1}^\prime \ \cdots \ \check{\textbf{H}}\tilde{\textbf{a}}_{j,T} \ \check{\textbf{H}}\tilde{\textbf{a}}_{j,T}^\prime ]^T
\end{equation*}
Since $\textbf{M}$ is real and skew-symmetric, $tr(\textbf{M}) =  0$. So,
\begin{eqnarray}
 \label{abc}
 \sum_{m =1}^{T}\{\langle \check{\textbf{H}}\tilde{\textbf{a}}_{j,m}, \check{\textbf{H}}\tilde{\textbf{a}}_{i,m} \rangle + \langle \check{\textbf{H}}\tilde{\textbf{a}}_{j,m}^\prime, \check{\textbf{H}}\tilde{\textbf{a}}_{i,m}^\prime \rangle \} & = & 0 \\ \label{bca}
 \Leftrightarrow 2\sum_{m =1}^{T}\langle \check{\textbf{H}}\tilde{\textbf{a}}_{j,m}, \check{\textbf{H}}\tilde{\textbf{a}}_{i,m} \rangle & = & 0 \\
 \label{hca}
 \therefore \sum_{m =1}^{T}\langle \check{\textbf{H}}\tilde{\textbf{a}}_{j,m}, \check{\textbf{H}}\tilde{\textbf{a}}_{i,m} \rangle & = & 0
\end{eqnarray}
where, \eqref{abc} follows from Lemma \ref{lemma1} and \eqref{bca} follows from \eqref{cda}. Now,
\begin{eqnarray*}
\textbf{H}_{eq} & = & \left(\textbf{I}_T \otimes \check{\textbf{H}}\right)\textbf{G} \\
&  = &  \left(\textbf{I}_T \otimes \check{\textbf{H}}\right)\left[ \begin{array}{cccc}
            \tilde{\textbf{a}}_{1,1} & \tilde{\textbf{a}}_{2,1} & \cdots & \tilde{\textbf{a}}_{2k,1} \\
            \tilde{\textbf{a}}_{1,2} & \tilde{\textbf{a}}_{2,2} & \cdots & \tilde{\textbf{a}}_{2k,2} \\
            \vdots & \vdots & \ddots & \vdots \\
            \tilde{\textbf{a}}_{1,T} & \tilde{\textbf{a}}_{2,T} & \cdots & \tilde{\textbf{a}}_{2k,T} \\
            \end{array}\right ] \\
&  =&   \left[ \begin{array}{cccc}
            \check{\textbf{H}}\tilde{\textbf{a}}_{1,1} & \check{\textbf{H}}\tilde{\textbf{a}}_{2,1} & \cdots & \check{\textbf{H}}\tilde{\textbf{a}}_{2k,1} \\
            \check{\textbf{H}}\tilde{\textbf{a}}_{1,2} & \check{\textbf{H}}\tilde{\textbf{a}}_{2,2} & \cdots & \check{\textbf{H}}\tilde{\textbf{a}}_{2k,2} \\
            \vdots & \vdots & \ddots & \vdots \\
            \check{\textbf{H}}\tilde{\textbf{a}}_{1,T} & \check{\textbf{H}}\tilde{\textbf{a}}_{2,T} & \cdots & \check{\textbf{H}}\tilde{\textbf{a}}_{2k,T} \\
            \end{array}\right ] \\
\end{eqnarray*}
From the above structure, it is readily seen that for any $i,j$, $i \neq j, 1 \leq i,j \leq 2k$, if $\textbf{A}_i\textbf{A}_j^H+\textbf{A}_j\textbf{A}_i^H = \textbf{O}_{n_t}$, then the $i^{th}$ and the $j^{th}$ columns of $\textbf{H}_{eq}$ are orthogonal. This follows from \eqref{hca}.
\end{proof}

Now, let us consider the proposed STBC for $2\times2$ MIMO system. Here, $k = 4, T = 2$. It can be verified that the following holds true for $l,m \in \left\{1,2,3,4\right\}$
\begin{equation} \label{cond}
 \textbf{A}_m\textbf{A}_l^H + \textbf{A}_l\textbf{A}_m^H = \textbf{O}_{n_t}  \left\{ \begin{array}{ll}
\forall l \neq m, m+1, & \textrm{if m is odd}\\
\forall l \neq m, m-1, & \textrm{if m is even}.\\
\end{array} \right.
\end{equation}
To be precise, \eqref{cond} holds for $(i,j) \in \{ (1,3),(1,4),(2,3),(2,4) \}$. Therefore, from Theorem \ref{thm1}, $\langle \textbf{h}_1,\textbf{h}_3 \rangle = \langle \textbf{h}_1,\textbf{h}_4 \rangle = \langle \textbf{h}_2,\textbf{h}_3 \rangle = \langle \textbf{h}_2,\textbf{h}_4 \rangle = 0 $.

Using the above results in the definition of the $\textbf{R}$-matrix, it can easily be shown that $\langle \textbf{q}_1,\textbf{h}_3 \rangle = \langle \textbf{q}_1,\textbf{h}_4 \rangle = \langle \textbf{q}_2,\textbf{h}_3 \rangle = \langle \textbf{q}_2,\textbf{h}_4 \rangle = 0 $. So, the structure of the $\textbf{R}$-matrix for our $2\times2$ code is as follows.
\begin{equation}
\label{formofR}
\textbf{R} =  \left[ \begin{array}{cccccccc}
\Vert \textbf{r}_1 \Vert & \langle \textbf{q}_1,\textbf{h}_2 \rangle  & 0 & 0 & \langle \textbf{q}_1,\textbf{h}_5 \rangle  & \langle \textbf{q}_1,\textbf{h}_6 \rangle  & \langle \textbf{q}_1,\textbf{h}_7 \rangle  & \langle \textbf{q}_1,\textbf{h}_8 \rangle  \\
0 & \Vert \textbf{r}_2 \Vert & 0 & 0 & \langle \textbf{q}_2,\textbf{h}_5 \rangle  & \langle \textbf{q}_2,\textbf{h}_6 \rangle  & \langle \textbf{q}_2,\textbf{h}_7 \rangle  & \langle \textbf{q}_2,\textbf{h}_8 \rangle  \\
0 & 0 & \Vert \textbf{r}_3 \Vert  & \langle \textbf{q}_3,\textbf{h}_4 \rangle  & \langle \textbf{q}_3,\textbf{h}_5 \rangle  & \langle \textbf{q}_3,\textbf{h}_6 \rangle  & \langle \textbf{q}_3,\textbf{h}_7 \rangle  & \langle \textbf{q}_3,\textbf{h}_8 \rangle  \\
0 & 0 & 0 & \Vert \textbf{r}_4 \Vert & \langle \textbf{q}_4,\textbf{h}_5 \rangle  & \langle \textbf{q}_4,\textbf{h}_6 \rangle  & \langle \textbf{q}_4,\textbf{h}_7 \rangle  & \langle \textbf{q}_4,\textbf{h}_8 \rangle  \\
0 & 0 & 0 & 0 & \Vert \textbf{r}_5 \Vert & \langle \textbf{q}_5,\textbf{h}_6 \rangle  & \langle \textbf{q}_5,\textbf{h}_7 \rangle  & \langle \textbf{q}_5,\textbf{h}_8 \rangle  \\
0 & 0 & 0 & 0 & 0 & \Vert \textbf{r}_6 \Vert & \langle \textbf{q}_6,\textbf{h}_7 \rangle  & \langle \textbf{q}_6,\textbf{h}_8 \rangle  \\
0 & 0 & 0 & 0 & 0 & 0 & \Vert \textbf{r}_7 \Vert & \langle \textbf{q}_7,\textbf{h}_8 \rangle  \\
0 & 0 & 0 & 0 & 0 & 0 & 0 & \Vert \textbf{r}_8 \Vert \end{array} \right]
\end{equation}

The structure of the ${\textbf{R}}$-matrix enables one to achieve simplified ML-decoding. This is because once the symbols $x_3$ and $x_4$ are given, $x_1$ and $x_2$ can be decoded independently. In the ML-decoding metric, it can be observed that the real and imaginary parts of symbol $x_1$ are entangled with one another but are independent of the real and imaginary parts of $x_2$ when $x_3$ and $x_4$ are conditionally given. So, the number of metric computations required is at most $(M^2)(2M) = 2M^3$ and hence, the ML-decoding complexity is of the order of $M^3$. When the constellation employed is a square QAM so that the real and the imaginary parts of each symbol can be decoded independently, the ML-decoding complexity can be further reduced as follows. Let $\hat{\textbf{x}} \triangleq [\hat{x}_{1I}, \hat{x}_{1Q}, \cdots, \hat{x}_{4Q}]$ denote the decoded information vector. Assuming that sphere decoding is employed (sphere decoding can be employed for constellations like  square or rectangular QAM and not for any arbitrary constellation which is a finite subset of $\mathbb{Z}[j]$), the following strategy is employed -
\begin{enumerate}
 \item A 4 dimensional real SD is done to decode the symbols $x_4$ and $x_3$, and there are $M^2$ such pairs for an M-QAM constellation.
 \item Next, $x_{2Q}$ is decoded in parallel with $x_{1Q}$, and there are  $\sqrt{M}$ possibilities for each of them. Following this, $x_{1I}$ and $x_{2I}$ are decoded using \emph{hard-limiting}, as follows
     \begin{equation}
      \hat{x}_{1I} = min\Big\{max\big(2rnd\big[\frac{u_1}{2}\big], -M\big), M\Big\}
     \end{equation}
     \begin{equation}
      \hat{x}_{2I} = min\Big\{max\big(2rnd\big[\frac{u_2}{2}\big], -M\big), M\Big\}
     \end{equation}
     where,
     \begin{equation*}
     u_1 \triangleq \big(y_1^\prime - r_{(1,2)}\hat{x}_{1Q} - \sum_{i=3}^{4}(r_{(1,2i-1)}\hat{x}_{iI}+ r_{(1,2i)}\hat{x}_{iQ})\big)/r_{(1,1)}
     \end{equation*}
     \begin{equation*}
     u_2 \triangleq \big(y_3^\prime - r_{(3,4)}\hat{x}_{2Q} - \sum_{i=3}^{4}(r_{(3,2i-1)}\hat{x}_{iI}+ r_{(3,2i)}\hat{x}_{iQ})\big)/r_{(3,3)}
     \end{equation*}
     \noindent where, for simplicity, we have denoted the $(i,j)^{th}$ entry of the $\textbf{R}$-matrix by $r_{(i,j)}$.
\end{enumerate}
So, the ML-decoding complexity of our code for square QAM is of the order of $M^2\sqrt{M}$. If, however, the QAM constellation used is not a square QAM, and cannot be represented as the Cartesian product of two PAM constellations (like the 32-QAM constellation, the optimum representation of which is as shown in Figure \ref{fig_non_square}), then the method described above cannot be employed. So, in such a scenario, the ML-decoding complexity becomes $M^3$, because one requires to decode wholly the complex symbols $x_1$ and $x_2$, when $x_3$ and $x_4$ are given.

Now, let us consider the proposed STBC for $4\times2$ MIMO system. For this case, $k = 8, T = 4$. It can be verified that the condition in \eqref{cond} holds true for $l,m \in \left\{1,2,\cdots,8\right\}$. Hence, from Theorem \ref{thm1}, for $l,m \in \left\{1,2,\cdots,8\right\},$ we have
\begin{equation*}
 \langle \textbf{h}_l, \textbf{h}_m \rangle = 0   \left\{ \begin{array}{ll}
\forall l \neq m, m+1, & \textrm{if m is odd}\\
\forall l \neq m, m-1, & \textrm{if m is even}.\\
\end{array} \right.
\end{equation*}
Using the above result, it can be easily be verified that for $l,m \in \{1,2,\cdots,8\}, l < m$,
\begin{equation*}
\langle \textbf{q}_l, \textbf{h}_m \rangle = 0   \left\{ \begin{array}{ll}
\forall l \neq m, m+1, & \textrm{if m is odd}\\
\forall l \neq m, m-1, & \textrm{if m is even}.\\
\end{array} \right.
\end{equation*}
For simplicity, let us define the $\textbf{R}$ matrix as follows
\begin{equation*}
\textbf{R} \triangleq \left[ \begin{array}{cc}
                \textbf{R}_1 & \textbf{R}_2 \\
                \textbf{O}_{8} & \textbf{R}_3 \\
                \end{array}\right]
\end{equation*}
\noindent where, $\textbf{R}_1, \textbf{R}_2$ and $\textbf{R}_3 \in \mathbb{R}^{8\times8}$, then, $\textbf{R}_1$ can be seen to have the following structure

\begin{equation}
\label{formofR1} \left[ \begin{array}{cccccccc}
\Vert \textbf{r}_1 \Vert & \langle \textbf{q}_1,\textbf{h}_2 \rangle  & 0 & 0 & 0 & 0 & 0 & 0 \\
0 & \Vert \textbf{r}_2 \Vert & 0 & 0 & 0 & 0 & 0 & 0 \\
0 & 0 & \Vert \textbf{r}_3 \Vert  & \langle \textbf{q}_3,\textbf{h}_4 \rangle  & 0 & 0 & 0 & 0 \\
0 & 0 & 0 & \Vert \textbf{r}_4 \Vert & 0 & 0 & 0 & 0 \\
0 & 0 & 0 & 0 & \Vert \textbf{r}_5 \Vert & \langle \textbf{q}_5,\textbf{h}_6 \rangle  & 0 & 0 \\
0 & 0 & 0 & 0 & 0 & \Vert \textbf{r}_6 \Vert & 0 & 0  \\
0 & 0 & 0 & 0 & 0 & 0 & \Vert \textbf{r}_7 \Vert & \langle \textbf{q}_7,\textbf{h}_8 \rangle  \\
0 & 0 & 0 & 0 & 0 & 0 & 0 & \Vert \textbf{r}_8 \Vert \end{array} \right].
\end{equation}
\noindent
The above structure of the matrix $\textbf{R}$ allows our code to achieve simplified ML-decoding as follows - having fixed the symbols $x_5,x_6,x_7$ and $x_8$, the symbols  $x_1,x_2,x_3$ and $x_4$ can be decoded independently. In the decoding metric, it can be observed that the real and imaginary parts of symbol $x_1$ are entangled with one another but are independent of the real and imaginary parts of $x_2$, $x_3$ and $x_4$ when $x_5,x_6,x_7$ and $x_8$ are conditionally given. Similarly, $x_2$, $x_3$ and $x_4$ are decoupled from one another although their own real and imaginary parts are coupled with one another. So, in general, the ML-decoding complexity of our code is of the order of $M^5$. That is due to the fact that jointly decoding the symbols $x_5,x_6,x_7$ and $x_8$ followed by independently decoding $x_1,x_2,x_3$ and $x_4$ in parallel requires a total of $(M^4)(4M)=4M^5$ metric computations. However, when square QAM is employed, the ML-decoding complexity can be further reduced as follows. Let $\hat{\textbf{x}} \triangleq [\hat{x}_{1I}, \hat{x}_{1Q}, \cdots, \hat{x}_{8Q}]$ denote the decoded information vector. Assuming the use of a sphere decoder,
\begin{enumerate}
 \item an 8 dimensional real SD is done to decode the symbols $x_5,x_6,x_7$ and $x_8$.
 \item Next, $x_{1Q}$, $x_{2Q}$, $x_{3Q}$ and $x_{4Q}$ are decoded in parallel. Following this, $x_{1I}$, $x_{2I}$ $x_{3I}$ and $x_{4I}$ are decoded using \emph{hard limiting} as follows
     \begin{equation}
      \hat{x}_{1I} = min\Big\{max\big(2rnd\big[\frac{u_1}{2}\big], -M\big), M\Big\}
     \end{equation}
     \begin{equation}
      \hat{x}_{2I} = min\Big\{max\big(2rnd\big[\frac{u_2}{2}\big], -M\big), M\Big\}
     \end{equation}
     \begin{equation}
      \hat{x}_{3I} = min\Big\{max\big(2rnd\big[\frac{u_3}{2}\big], -M\big), M\Big\}
     \end{equation}
     \begin{equation}
      \hat{x}_{4I} = min\Big\{max\big(2rnd\big[\frac{u_4}{2}\big], -M\big), M\Big\}
     \end{equation}
     \end{enumerate}
     where,
     \begin{equation*}
     u_1 \triangleq \big(y_1^\prime - r_{(1,2)}\hat{x}_{1Q} - \sum_{i=5}^{8}(r_{(1,2i-1)}\hat{x}_{iI}+ r_{(1,2i)}\hat{x}_{iQ})\big)/r_{(1,1)}
     \end{equation*}
     \begin{equation*}
     u_2 \triangleq \big(y_3^\prime - r_{(3,4)}\hat{x}_{2Q} - \sum_{i=5}^{8}(r_{(3,2i-1)}\hat{x}_{iI}+ r_{(3,2i)}\hat{x}_{iQ})\big)/r_{(3,3)}
     \end{equation*}
     \begin{equation*}
     u_3 \triangleq \big(y_5^\prime - r_{(5,6)}\hat{x}_{3Q} - \sum_{i=5}^{8}(r_{(5,2i-1)}\hat{x}_{iI}+ r_{(5,2i)}\hat{x}_{iQ})\big)/r_{(5,5)}
     \end{equation*}
     \begin{equation*}
     u_4 \triangleq \big(y_7^\prime - r_{(7,8)}\hat{x}_{4Q} - \sum_{i=5}^{8}(r_{(7,2i-1)}\hat{x}_{iI}+ r_{(7,2i)}\hat{x}_{iQ})\big)/r_{(7,7)}
     \end{equation*}
     \noindent where, $r_{(i,j)}$ denotes the $(i,j)^{th}$ entry of the $\textbf{R}$-matrix.
So, in all, we need to make a maximum of $4M^4\sqrt{M}$ metric computations only. Hence, for square QAM constellations, the ML-decoding complexity of our code is of the order of $M^4\sqrt{M}$.

\section{Comparison of Ml-decoding complexity of our codes with  known $2\times2$ and $4\times2$ STBCs}
\label{mlexist}

The ML-decoding complexity of our $2\times2$ code was shown in the previous section to be of the order of $M^3$. This was due solely to the behavior of the weight matrices which resulted in the $\textbf{R}$-matrix structure as in \eqref{formofR} for our $2\times2$ code. For any code, the weight matrices entirely define the ML-decoding complexity. For eg., all the weight matrices of the Alamouti code satisfy the condition in Theorem \ref{thm1}, and hence, the equivalent channel matrix $\textbf{H}_{eq}$ is orthogonal. So, the $\textbf{R}$-matrix for the Alamouti code is a diagonal matrix and this results in an ML-decoding complexity of the order of $M$ for general constellations. In the special case of the constellation being a square QAM constellation, the real and imaginary parts of each symbol can be further decoded independently using hard-limiting and the decoding complexity of the Alamouti code for square $M$-QAM constellations is constant. For the Golden code, the ML-decoding complexity has been considered to be of the order of $M^4$ in the literature \cite{PGA}, \cite{SS},\cite{BHV}. However, the ML-decoding complexity of the Golden code can be reduced to the order of $M^2\sqrt{M}$ for square QAM constellations. It can be easily verified, by studying the weight matrices and using Theorem \ref{thm1}, that the Golden code has the following $\textbf{R}$-matrix structure:-
 \begin{equation*}
 \textbf{R}_{Golden~code} =  \left[ \begin{array}{rrrrrrrr}
         a & 0 & a & 0 & a & a & a & a\\
         0 & a & 0 & a & a & a & a & a\\
         0 & 0 & a & 0 & a & a & a & a\\
         0 & 0 & 0 & a & a & a & a & a\\
         0 & 0 & 0 & 0 & a & a & a & a\\
         0 & 0 & 0 & 0 & 0 & a & a & a\\
         0 & 0 & 0 & 0 & 0 & 0 & a & a\\
         0 & 0 & 0 & 0 & 0 & 0 & 0 & a\\
        \end{array}\right]
        \end{equation*}
        where '$a$' denotes a possible non-zero entry. This structure makes the ML-decoding complexity of the Golden code evident. In general, the ML-decoding complexity is of the order of $M^4$. However, when square $M$-QAM is employed, the following decoding strategy can be employed, assuming that a sphere decoder is used.
        \begin{enumerate}
        \item A 4-dimensional real SD is done to decode the symbols $x_4$ and $x_3$.
        \item Next, $x_{2I}$ and $x_{2Q}$ are decoded in parallel. Following this, $x_{1I}$ and $x_{1Q}$ are decoded as follows
            \begin{equation}
      \hat{x}_{1I} = min\Big\{max\big(2rnd\big[\frac{u_1}{2}\big], -M\big), M\Big\}
     \end{equation}
     \begin{equation}
      \hat{x}_{1Q} = min\Big\{max\big(2rnd\big[\frac{u_2}{2}\big], -M\big), M\Big\}
     \end{equation}
     where,
     \begin{equation*}
     u_1 \triangleq \big(y_1^\prime - r_{(1,3)}\hat{x}_{2I} - \sum_{i=3}^{4}(r_{(1,2i-1)}\hat{x}_{iI}+ r_{(1,2i)}\hat{x}_{iQ})\big)/r_{(1,1)}
     \end{equation*}
     \begin{equation*}
     u_2 \triangleq \big(y_2^\prime - r_{(2,4)}\hat{x}_{2Q} - \sum_{i=3}^{4}(r_{(3,2i-1)}\hat{x}_{iI}+ r_{(3,2i)}\hat{x}_{iQ})\big)/r_{(2,2)}
     \end{equation*}
     \noindent where, as usual, we have denoted the $(i,j)^{th}$ entry of the $\textbf{R}$-matrix by $r_{(i,j)}$ and $\hat{\textbf{x}} \triangleq $ $ [\hat{x}_{1I},\hat{x}_{1Q},\cdots,\hat{x}_{4Q}]$ denotes the decoded information vector.
    \end{enumerate}
         So, the ML-decoding complexity is of the order of $M^2\sqrt{M}$, the same as that for our $2\times2$ code. However, for non-rectangular QAM constellations, the Golden code does not admit simplified ML-decoding. The codes presented in \cite{DV}, \cite{JJK} and \cite{YW} also have their $\textbf{R}$-matrix structures identical to that of the Golden code and hence offer the same ML-decoding complexity.

Considering the HTW-PGA code, the $\textbf{R}$-matrix structure is observed to be as follows:-
\begin{equation*}
\textbf{R}_{HTW-PGA} =  \left[ \begin{array}{rrrrrrrr}
         a & 0 & 0 & 0 & a & a & a & a\\
         0 & a & 0 & 0 & a & a & a & a\\
         0 & 0 & a & 0 & a & a & a & a\\
         0 & 0 & 0 & a & a & a & a & a\\
         0 & 0 & 0 & 0 & a & 0 & 0 & 0\\
         0 & 0 & 0 & 0 & 0 & a & 0 & 0\\
         0 & 0 & 0 & 0 & 0 & 0 & a & 0\\
         0 & 0 & 0 & 0 & 0 & 0 & 0 & a\\
        \end{array}\right]
\end{equation*}
where '$a$' again denotes a possible non-zero entry. From this structure, the order of the ML-decoding complexity can be easily calculated for the different QAM  constellation types. For square $M$-QAM, it is of the order of  $M^2$. This follows from the fact that when the symbols $x_3$ and $x_4$ are fixed, $x_{1I}$, $x_{1Q}$, $x_{2I}$ and $x_{2Q}$ can be decoded independently from one another and each of them can be decoded by using hard-limiting, hence requiring a total of only 4$M^2$ computations. For non-rectangular QAM constellations, the ML-decoding complexity is of the order of $2M^3$. The Sezginer-Sari code also has a similar ML-decoding complexity. The above observations are all captured in Tables \ref{table_11}. In the table, the ML-decoding complexity given for each code is the maximum number of metric computations needed.

The ML-decoding complexity of our $4\times2$ code was shown to be of the order of $M^5$ for general constellations, and $M^4\sqrt{M}$ for square QAM constellations. This simplified complexity was facilitated by the structure of the $\textbf{R}$-matrix, a part of which had the structure as in \eqref{formofR1}. The ML decoding complexity of the DjABBA code is of the order of $M^7$ in general, and of the order of $M^6$ for square $M$-QAM. To the  best of our knowledge, this hasn't been mentioned in literature. To see this, one has to look at the $\textbf{R}$-matrix structure for the DjABBA code which, as mentioned before, is dictated by the weight matrices for the code. The structure of the $\textbf{R}_1$-matrix for the DjABBA code, one corresponding to \eqref{formofR1}, is as follows
\begin{equation*}
 \textbf{R}_{1,DjABBA} = \left[ \begin{array}{rrrrrrrr}
         a & 0 & 0 & 0 & a & a & a & a\\
         0 & a & 0 & 0 & a & a & a & a\\
         0 & 0 & a & 0 & a & a & a & a\\
         0 & 0 & 0 & a & a & a & a & a\\
         0 & 0 & 0 & 0 & a & a & a & a\\
         0 & 0 & 0 & 0 & 0 & a & a & a\\
         0 & 0 & 0 & 0 & 0 & 0 & a & a\\
         0 & 0 & 0 & 0 & 0 & 0 & 0 & a\\
        \end{array}\right]
\end{equation*}
where '$a$' corresponds to a possible non-zero entry. For square $M$-QAM, it is evident that $x_{1I}$, $x_{1Q}$, $x_{2I}$ and $x_{2Q}$ can be decoded independently from one another, by using hard-limiting, when the symbols $x_3$, $x_4$, $x_5$, $x_6$, $x_7$ and $x_8$ are fixed. This allows an ML-decoding complexity of the order of $M^6$, with $M^6$ metric computations for decoding the other 6 symbols. This scheme can be employed only for square QAM constellations, so that the real and the imaginary parts can be decoded independently. However, for non-rectangular QAM constellations, one must decode $x_1$ and $x_2$ independently, when the rest of the symbols are given. So, the ML-decoding complexity is of the order of $M^7$.

The BHV code, which was designed primarily for simplified ML-decoding complexity, has a complexity of the order of $M^6$ in general and of the order of $M^4\sqrt{M}$ specifically for square $M$-QAM (Incidentally, the authors of \cite{BHV} haven't claimed this!). This follows from the structure of the  $\textbf{R}_1$-matrix as shown in \eqref{rbhv}, with $a$ denoting a possible non-zero entry.

For square $M$-QAM, the following strategy can be employed to decode the symbols.
 \begin{enumerate}
 \item An 8-dimensional real SD is employed to decode the symbols $x_5$, $x_6$, $x_7$ and $x_8$.
 \item Following this, $x_{3I}$, $x_{3Q}$, $x_{4I}$ and $x_{4Q}$ are decoded in parallel. Next, $x_{1I}$, $x_{1Q}$, $x_{2I}$ and $x_{2Q}$ are decoded by employing hard-limiting.
 \end{enumerate}
 Hence, the ML-decoding complexity of the BHV code is of the order of $M^4\sqrt{M}$, because a maximum of $4M^4\sqrt{M}$ metric computations need to be done in minimizing the ML-decoding metric. But for non-rectangular QAM constellations, the pairs $(x_1,x_3)$ and $(x_2,x_4)$ have to be decoded in parallel after jointly decoding the last four symbols, thus accounting for an ML-decoding complexity of the order of $M^6$. Table \ref{table22} captures the ML-decoding complexities for the three codes for the different classes of QAM constellations.
\begin{equation}\label{rbhv}
 \textbf{R}_{1,BHV} = \left[ \begin{array}{rrrrrrrr}
         a & 0 & 0 & 0 & a & 0 & 0 & 0\\
         0 & a & 0 & 0 & 0 & a & 0 & 0\\
         0 & 0 & a & 0 & 0 & 0 & a & 0\\
         0 & 0 & 0 & a & 0 & 0 & 0 & a\\
         0 & 0 & 0 & 0 & a & 0 & 0 & 0\\
         0 & 0 & 0 & 0 & 0 & a & 0 & 0\\
         0 & 0 & 0 & 0 & 0 & 0 & a & 0\\
         0 & 0 & 0 & 0 & 0 & 0 & 0 & a\\
        \end{array}\right]
\end{equation}

\section{Simulation Results}
\label{dis}
In all the simulation scenarios in this section,  we consider quasi-static Rayleigh flat fading channels and the plots are shown for the codeword error rate (CER) as a function of the  SNR at each receive antenna.
\subsection{$2\times2$ MIMO}
 Figure \ref{4qam} shows the CER performances of our $2\times2$ code, the Golden code and the HTW-PGA code, with all the codes employing the 4 QAM constellation.  Figure \ref{16qam} shows the CER plots for the three codes, with the constellation used being 16 QAM. In both the plots, we see that the CER curve for our $2\times2$ code is indistinguishable from that of the Golden code and this is due to the identical coding gains of the two codes. The HTW-PGA code has a slightly worse performance because of its lower coding gain. Table \ref{table1} gives a comparison between the minimum determinants of some well known $2\times2$ codes. It is to be noted that in obtaining the minimum determinants for these codes, we have ensured that the average energy per codeword is uniform across all codes, but the average energy per constellation has been allowed to increase with constellation size, or in other words, the average constellation energies haven't been normalized to unity.

\subsection{$4\times2$ MIMO}
Figure \ref{4qamnew} shows the CER performance plots for our $4\times2$ code, the well known DjABBA code \cite{HTW} and the BHV code \cite{BHV}, with all the codes using the 4-QAM constellation. Figure \ref{16qamnew} shows the CER performance for 16 QAM. Both the plots exhibit a similar trend, with our $4\times2$ code outperforming both the DjABBA code and the BHV code at high SNR, and  the DjABBA code in turn outperforming the BHV code. This can be attributed to the superior coding gain of our $4\times2$ code. The bad performance of the BHV code at a high SNR is due mainly to the fact that it does not have full-diversity. Table \ref{table22} gives a comparison between the minimum determinants of the above three codes. The minimum determinants of our $4\times2$ code for 4-QAM and 16-QAM has been calculated using exhaustive search and the constellation energy hasn't been normalized to unity. However it has been ensured that the average energy per codeword has been maintained uniform for all the three codes. The DjABBA code that we have used for our simulations is the one that has been optimized for performance, and proposed in Chapter $9$ of \cite{HTW}, . It can be seen that our code has a coding gain twice that of the DjABBA's.

\section{Concluding Remarks}
\label{conc}
In this paper, we have seen that it is possible to have full-rate codes with  simplified ML-decoding complexity without having to sacrifice performance. We presented two codes, one each for the $2\times2$ and the $4\times2$ MIMO system, both of which have lower ML-decoding complexity for general QAM constellations than the best known codes for such systems. Moreover, our $4\times2$ code outperforms the best DjABBA code while our $2\times2$ code matches the Golden code in performance. We also saw that the weight matrices play a decisive role in defining the ML-decoding complexity of an STBC and went on to show that some existing codes also offer simplified ML-decoding for square QAM constellations, something which was not known hitherto. Noting the similarity between the constructions of the $2\times2$ code and the $4\times2$ code, it is natural to see if the design procedure can be extended to $2^a$ transmit antennas, $a >2$. However, there are two main issues to be concerned about:
\begin{enumerate}
 \item For our $2\times2$ code, we showed analytically that the minimum determinant for regular QAM constellations is $3.2$. However, for our $4\times2$ code, we have checked that the minimum determinant for 4 and 16 QAM is 10.24 through exhaustive computer search. We couldn't do the same for higher constellation sizes, because such a search would run for weeks!.
     The rate of a square CIOD for $2^a$ transmit antennas is $\frac{a}{2^{a-1}}$, so that this STBC has $2a$ independent information symbols. If we were to extend our approach to $2^a$ transmit antennas, $a > 2$, the code would have $4a$ symbols and finding out the minimum determinant for 4 QAM itself would be time consuming.
 \item The ML-decoding complexity for our $2\times2$ code is of the order of $M^3$ and that for our $4\times2$ code is $M^5$, for general constellations. So, the ML-decoding complexity for the STBC designed for $2^a$ transmit antennas, $a >2$ would be of the order of $M^{2a+1}$, while the rate would be $\frac{a}{2^{a-2}}$. While there is an increase in code rate, there is also a substantial increase in ML-decoding complexity, making the attractiveness of code design using this approach for higher number of transmit antennas questionable.
\end{enumerate}

The following questions still remain unanswered. 
\begin{itemize}
\item For a $2\times2$ MIMO system, what is the minimum ML-decoding complexity achievable for a full-rate, full-diversity STBC ? Is it possible to have a full-rate, full-diversity code with an ML-decoding complexity of the order of  $M^2$ for all constellations.
\item Multi-group decodable codes \cite{sush} offer simplified ML-decoding complexity. For a given transmit antenna, what is the maximum rate that a multi-group decodable code can have ? For the $4\times2$ MIMO case, is it possible to have a full-rate, full-diversity, two-group decodable STBC, so that the ML-decoding complexity is of the order of $M^4$ ?
\end{itemize}

\section*{ACKNOWLEDGEMENTS}
This work was supported partly by the DRDO-IISc program on Advanced Research in Mathematical Engineering. We would also like to thank Mr. Shashidhar, working with Beceem Communications, for his useful discussion on the ML-decoding complexity issues.

\newpage
\begin{table*}
\begin{center}
\begin{tabular}{|c|c|} \hline
 & Min det   \\
 Code & for M QAM    \\ \hline

Tilted QAM \cite{YW} &  0.8000   \\ \hline
Dayal-Varanasi code \cite{DV} &  3.2000   \\ \hline
The Golden code \cite{BRV} &  3.2000   \\ \hline
Trace-orthonormal cyclotomic code \cite{JJK} &  3.2000  \\ \hline
Paredes-Gershman code \cite{PGA} &  2.2857  \\ \hline
Serdar-Sari code \cite{SS} & 2.0000  \\ \hline
The proposed code \cite{Pav1} &  3.2000 \\ \hline
\end{tabular}
\caption{Comparison between the minimum determinants of some well known $2\times2$ STBCs}
\label{table1}
\end{center}
\end{table*}

\begin{table*}
\begin{center}
\begin{tabular}{|c|c|c|} \hline
   & \multicolumn{2}{c|}{ ML Decoding complexity}  \\ \cline{2-3}
 Code  &  square QAM & Non-rectangular QAM  \\ \hline
Tilted QAM  & $2M^2\sqrt{M}$ &  $M^4$ \\ \hline
Dayal-Varanasi code   & $2M^2\sqrt{M}$ &  $M^4$ \\ \hline
The Golden code  & $2M^2\sqrt{M}$ &  $M^4$ \\ \hline
Trace-orthonormal cyclotomic code  & $2M^2\sqrt{M}$ &  $M^4$ \\ \hline
Paredes-Gershman code   & 4$M^2$ &  $2M^3$ \\ \hline
Serdar-Sari code    & 4$M^2$ &  $2M^3$ \\ \hline
The proposed code  & $2M^2\sqrt{M}$ & $2M^3$ \\ \hline
\end{tabular}
\caption{Comparison between the ML-decoding complexity of some well known $2\times2$ STBCs for QAM}
\label{table_11}
\end{center}
\end{table*}

\begin{table*}
\begin{center}
\begin{tabular}{|c|c|c|c|} \hline
    & Min det for &  \multicolumn{2}{c|}{ ML Decoding complexity}  \\ \cline{3-4}
 Code & 4 and 16 QAM & Square QAM &  Non-rectangular QAM  \\ \hline
DjABBA code \cite{HTW} & 0.64 & 4$M^6$ &  $2M^7$ \\ \hline
BHV code \cite{BHV} & 0 & 4$M^4\sqrt{M}$ &  $2M^6$ \\ \hline
The proposed code \cite{Pav3} & 10.24 &  $4M^4\sqrt{M} $ & $4M^5$ \\ \hline
\end{tabular}
\end{center}
\caption{Comparison between the minimum determinant and the ML-decoding complexity of  $4\times2$ STBCs for QAM constellations}
\label{table22}
\end{table*}


\begin{figure}
\centering
\includegraphics[width=5.5in,height=3.8in]{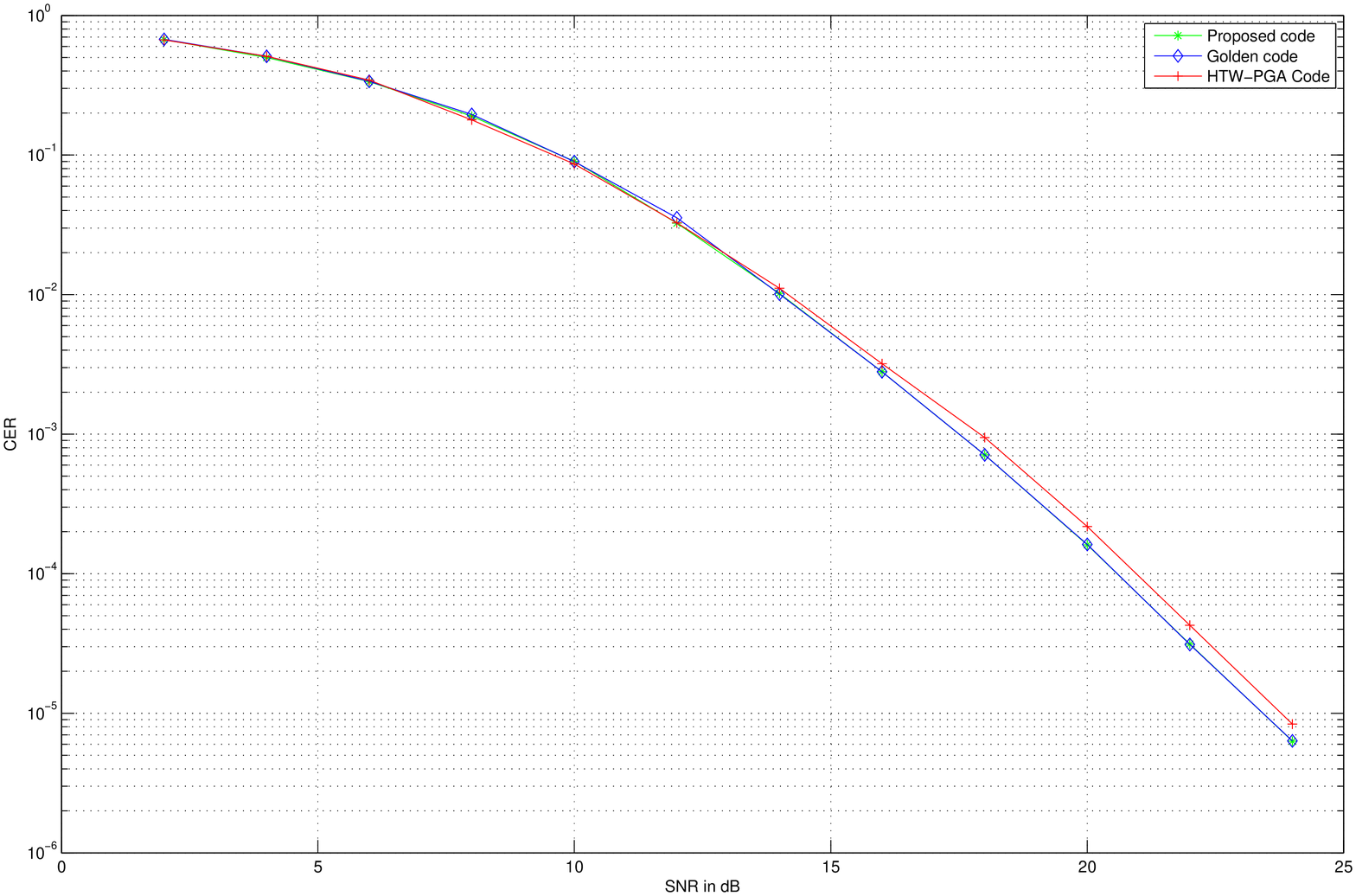}
\caption{CER performance of $2\times2$ codes for 4-QAM}
\label{4qam}
\end{figure}

\begin{figure}
\centering
\includegraphics[width=5.5in,height=3.8in]{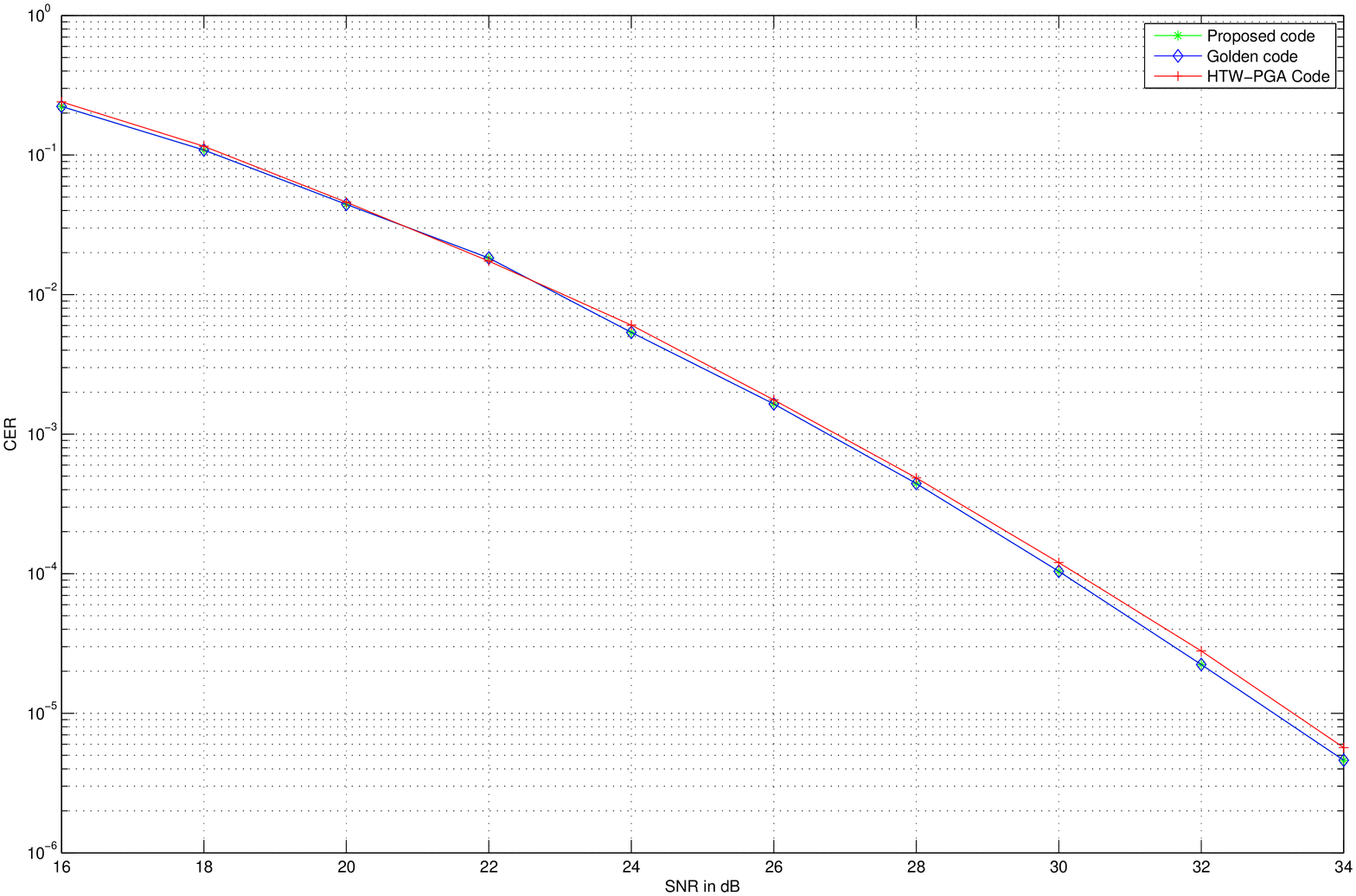}
\caption{CER performance of $2\times2$ codes for 16-QAM}
\label{16qam}
\end{figure}

\begin{figure}
\centering
\includegraphics[width=5.5in,height=3.8in]{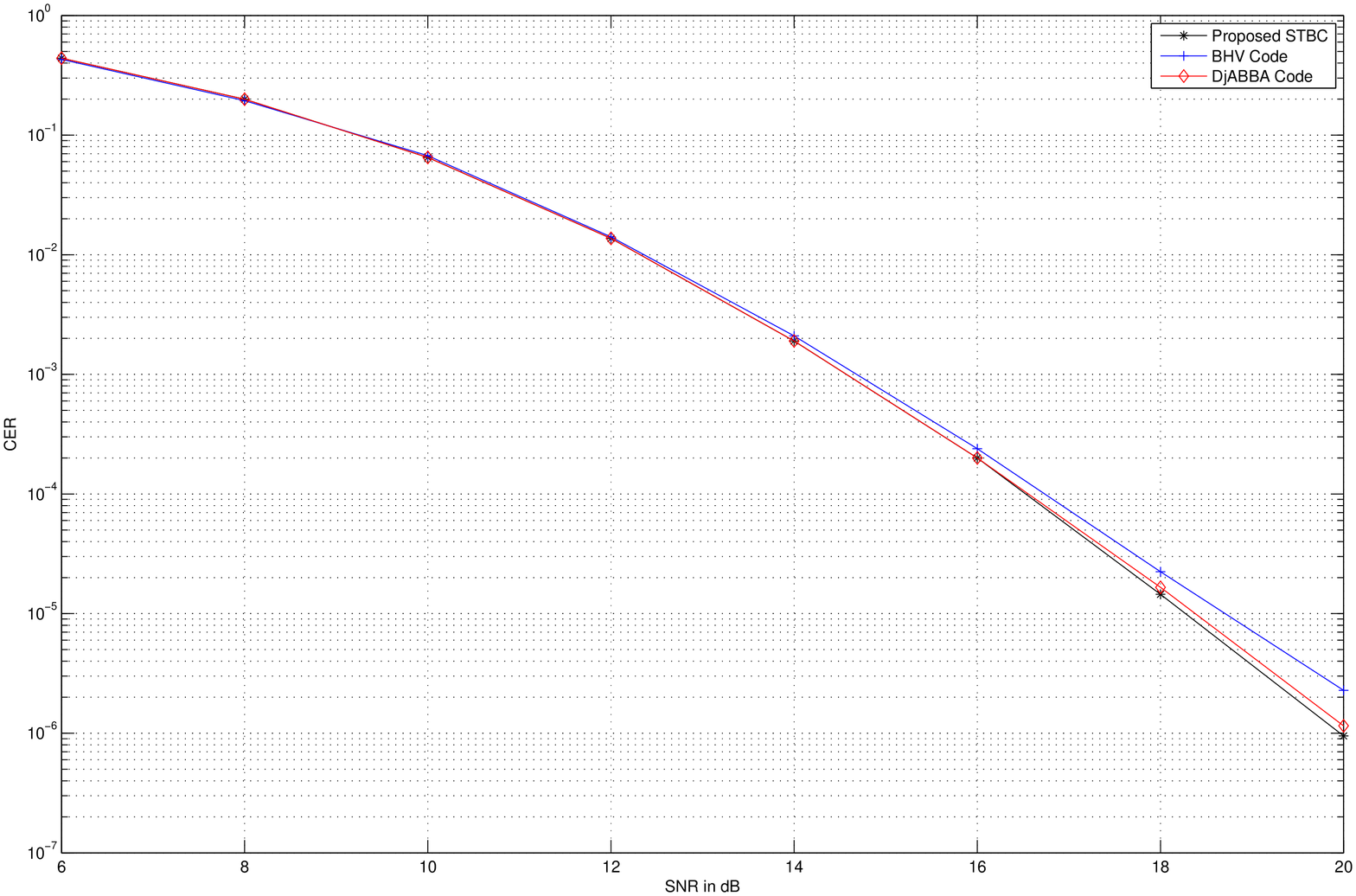}
\caption{CER performance of $4\times2$ codes for 4-QAM}
\label{4qamnew}
\end{figure}

\begin{figure}
\centering
\includegraphics[width=5.5in,height=3.8in]{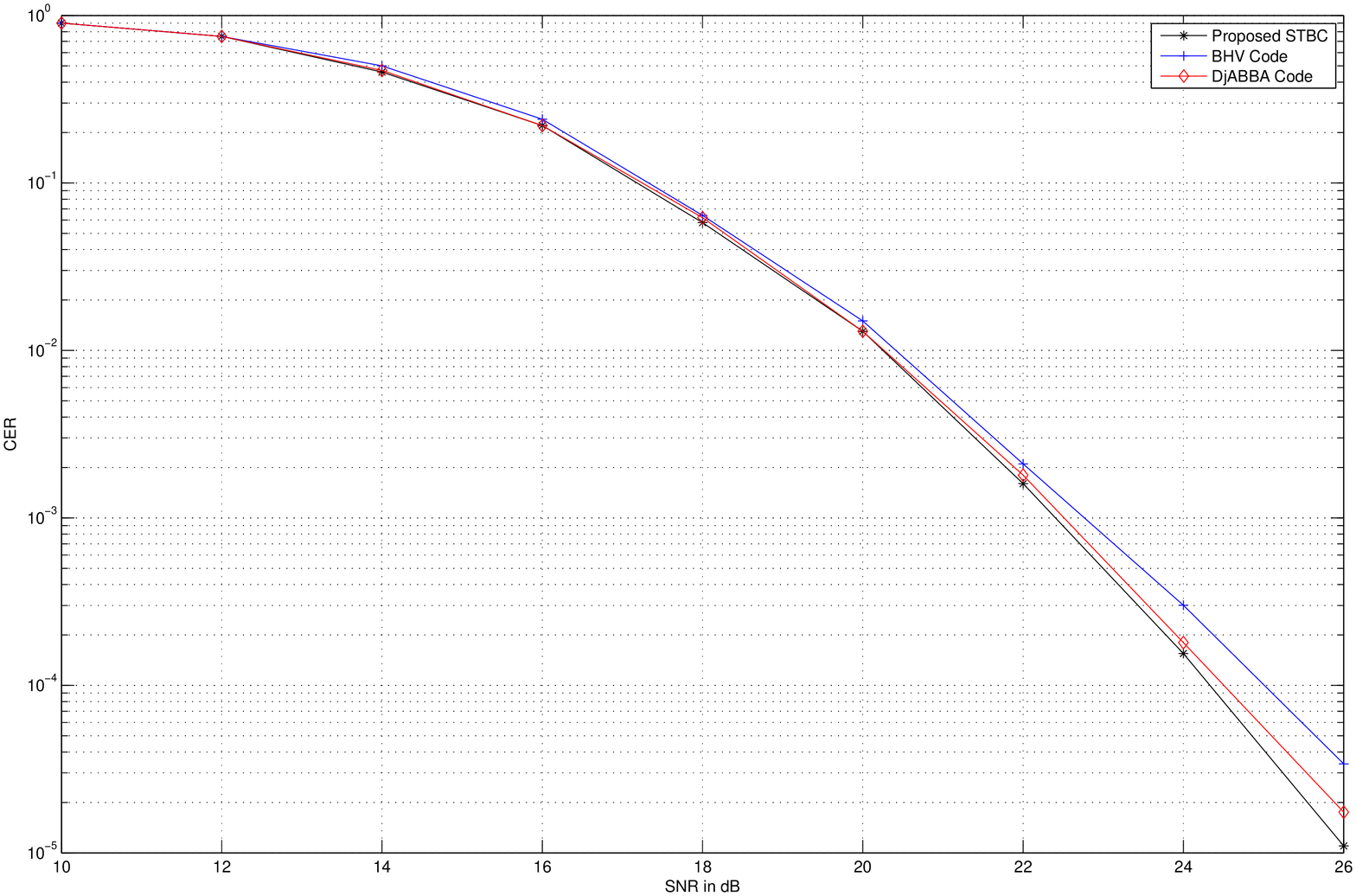}
\caption{CER performance for $4\times2$ codes for 16-QAM}
\label{16qamnew}
\end{figure}

\begin{figure}
\centering
\includegraphics[width=5.5in,height=3.8in]{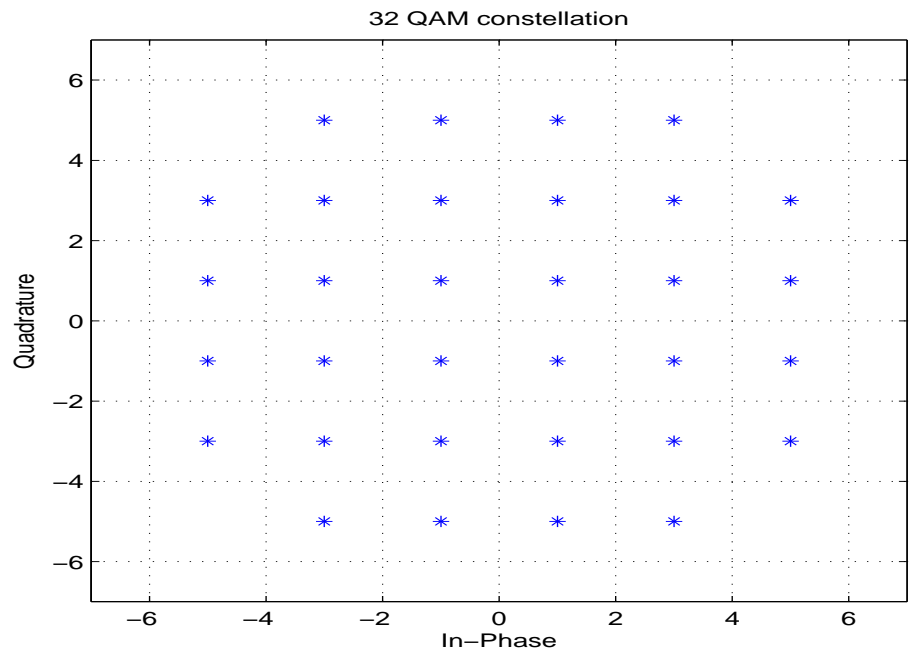}
\caption{32 QAM - An example of a non-rectangular QAM constellations}
\label{fig_non_square}
\end{figure}

\end{document}